\documentclass[conference]{IEEEtran}
\IEEEoverridecommandlockouts
\usepackage{caption}
\usepackage{amsmath,amssymb,amsthm}
\usepackage{algorithm}
\usepackage{algorithmicx}
\usepackage{algpseudocode}
\usepackage{enumitem}
\usepackage[hang,flushmargin]{footmisc}
\usepackage{bbm}
\newtheoremstyle{mytheoremstyle}
  {\topsep} 
  {\topsep} 
  {\itshape} 
  {} 
  {\bfseries} 
  {.} 
  {.5em} 
  {} 
\theoremstyle{mytheoremstyle}
\newtheorem{proposition}{Proposition}
\newtheorem{lemma}{Lemma}
\newtheorem{theorem}{Theorem}

\usepackage{cite}
\usepackage{amsmath,amssymb,amsfonts}
\usepackage{amsthm}
\usepackage{graphicx}
\usepackage{textcomp}
\usepackage{xcolor}
\usepackage{algorithm,algorithmicx}
\usepackage{booktabs} 
\usepackage{multirow}
\usepackage{subcaption}
\usepackage{siunitx}
\usepackage{url}
\def\BibTeX{{\rm B\kern-.05em{\sc i\kern-.025em b}\kern-.08em
    T\kern-.1667em\lower.7ex\hbox{E}\kern-.125emX}}
\renewcommand{\baselinestretch}{0.955}   
\usepackage[top=0.73in,bottom=1in,left=0.6in,right=0.6in]{geometry}
\begin{document}
\title{Discrete Time  Credit-Based Shaping for Time-Sensitive Applications in  5G/6G Networks}
\author{Anudeep Karnam, Kishor C. Joshi,  Jobish John, George Exarchakos, Sonia Heemstra de Groot, Ignas Niemegeers\\
Center for Wireless Technology Eindhoven, Eindhoven University of Technology, The Netherlands.}
\maketitle
\begin{abstract}
Credit-Based Shaping (CBS) is one of the key mechanisms used in Time-Sensitive Networking (TSN) to provide bounded latency for prioritized traffic on wired Ethernet. However, a direct port of CBS to 5G/6G New Radio (NR) is not viable because of slotted time, centralized grants, and Modulation and Coding Scheme (MCS)-dependent rates. We introduce a slot-native, per-UE credit gate inside the gNB Medium Access Control (MAC) that restores TSN-like latency isolation over the air. Our first design, \emph{CBS-DT}, updates credit once per slot and debits each new grant by its Transport Block Size~(TBS). To avoid over-penalizing partially used grants (padding, burstiness), \emph{CBS-PU} debits only delivered bytes. For both variants we establish bounded credit and deterministic deficit recovery; with Round-Robin (RR) over the set of eligible UEs and at most \(K\) new grants per slot we derive hard bounds on (i) time to regain eligibility, (ii) time to first service once eligible, and (iii) inter-grant separation. We further develop an event-driven realization that replaces per-slot \(\mathcal{O}(U)\) scans with wake-ups, achieving total work \(\mathcal{O}(N{+}G)+\mathcal{O}((A{+}G)\log U)\) over \(N\) slots, where \(U\) is the UE count and \(A,G\) are queue-activation and new-grant events, so cost tracks traffic rather than population. In a 3GPP-conformant MATLAB 5G Toolbox evaluation, both variants preserve strict class ordering under heavy load and tighten latency tails versus proportional-fair (PF) and weighted-PF (WPF) schedulers. At reduced load, CBS-PU eliminates padding over-debit and attains \(>\!98\%\) grant utilization, outperforming the base RR policy while preserving isolation. The credit gate requires no changes to Hybrid Automatic Repeat reQuest (HARQ) or air-interface signaling, offering a practical path to deterministic 5G/6G networks.
\end{abstract}
\section{Introduction}
\label{sec_introduction}
Beyond-5G (B5G) and 6G networks are expected to support a wide range of quality-of-service (QoS) requirements~\cite{6G_vision}. Emerging Industry~$4.0/5.0$ use cases such as smart factories and real-time motion control require end-to-end latencies below $1~\text{ms}$ with microsecond-level jitter, and reliabilities exceeding $99.999\%$~\cite{Industry5O}. On wired links, the IEEE Time-Sensitive Networking (TSN) toolset provides such determinism by combining per-flow shaping with time-aware queuing to yield hard delay bounds~\cite{Sasiain}. Achieving comparable guarantees over the stochastic wireless medium is challenging due to time-varying channel quality, grant-based scheduling, and retransmissions~\cite{Sharma}. A natural approach is to adapt TSN primitives to 5G/6G wireless networks, notably the IEEE 802.1Qav~\cite{standard} credit-based shaping (CBS), which achieves per-class bounded delay when the aggregate reserved rate is strictly below the link capacity~\cite{fi16090345}.  This will extend predictable, bounded latency to high-priority traffic,  while preventing starvation of lower-priority traffic.

 TSN-over-5G deployments introduce device-side and network-side TSN translators (DS-TT/NW-TT) that map IEEE~802.1Q semantics onto 5G QoS identifier  (5QI)~\cite{Nikhileswar,Atiq}.   Complementary efforts build TSN–5G testbeds and virtualized environments~\cite{garbugli1,garbugli12}. These systems enable end-to-end TSN connectivity but treat the 5G cell as a transparent bridge~\cite{9212141} and leave the NR medium access control (MAC) scheduler unmodified. Closer to the scheduler, configured-grant (CG) approaches coordinate 5G grants with TSN flow periodicity to increase the fraction of flows meeting latency targets~\cite{LARRANAGAZUMETA2024107930}, coupling TSN knowledge to grant timing. Wireless TSN surveys flag \emph{slot-level scheduling under variable MCS} as an open problem~\cite{Zanbouri_2025}. Within ultra-reliable low-latency communications (URLLC), mechanisms such as grant-free access, mini-slots, and HARQ tuning improve latency but can starve low-priority flows when critical flows are over-provisioned~\cite{URRLCsurvey2025,Maghsoudnia}. Even with DS-TT/NW-TT that map Ethernet priorities into 5QI, the gNB MAC continues to run the standard, non-credit-aware grant scheduler~\cite{Satka}. As a result, CBS-derived determinism can be lost inside the NR air interface, where bursts accumulate and slot-level pacing is not enforced. What is missing is a \emph{slot-native} credit mechanism inside the gNB that debits credit whenever a new transmission grant is issued. We address this gap with a credit gate embedded in the gNB: per-UE credit evolves in discrete time, debits on each issued grant, and can be tied either to granted transport block size (TBS) or to the actually delivered bytes (to avoid over-debiting under partial usage). This enables TSN-like latency isolation in NR while preserving spectral efficiency. Our main contributions can be summarized as follows:
\begin{enumerate}[leftmargin=*]
\item \textit{Discrete time CBS (CBS-DT):}
We embed a per-UE credit gate in the gNB MAC and update credit on slot boundaries, debiting each \emph{new} transmission by the granted TBS. This aligns TSN credit semantics with NR’s slotted, MCS-dependent service (Section~\ref{sec_cbsstd}). We establish per-UE \emph{bounded credit} via clamping and \emph{deterministic recovery} from any deficit with a finite recovery time  (Section~\ref{sec:complexity}).
\item \textit{Partial-usage debit (CBS-PU):}
To avoid over-penalizing partially used grants, we debit by \emph{delivered} bytes. We prove \emph{credit dominance}, \emph{rate preservation}, and \emph{no-worse delay} versus CBS-DT (Theorem~\ref{thm:pu_properties}).
\item \textit{Deterministic pacing and delay bounds:}
Under at most $K$ new grants per slot and round-robin (RR) policy over the set of eligible UEs for transmission, we derive hard bounds on: (i) time to eligibility from any deficit; (ii) time to first service once eligible;  and (iii) inter-grant separation (Lemma~\ref{lem:absolute-wait}).
\item \textit{Event-driven implementation and scalability:}
We replace per-slot $\mathcal{O}(U)$ scans with wake-up scheduling. Over $N$ slots the total work is $\mathcal{O}\!\big(N+G\big)+\mathcal{O}\!\big((A{+}G)\log U\big)$, where $U$ is the total number of UEs, and $A$ and $G$ are queue-activation and new-grant events. When event rates do not grow with $U$, the scheduler’s cost \emph{tracks traffic, not population}.
\item \textit{3GPP-conformant evaluation:}
Using MATLAB \emph{5G Toolbox} (TS~38.211/38.212/38.214/38.322), we compare CBS-DT/CBS-PU (with RR over eligible UEs) to proportional fair (PF) and weighted-PF (WPF) in a single-cell NR downlink (Section~\ref{sec_results}). Both CBS variants preserve class ordering and tighten latency tails.
\end{enumerate}
\noindent \subsubsection*{Organization}
Section~\ref{sec:background_cbs} reviews TSN CBS. Section~\ref{sec:system_model} gives the NR system model. Sections~\ref{sec_cbsstd} and \ref{sec_cbspu} develop CBS-DT and CBS-PU and their properties. Section~\ref{sec:complexity} presents the event-driven algorithm and complexity. Section~\ref{sec_results} reports evaluation results; Section~\ref{sec_conclusion} concludes.
\section{Background: Credit-Based Shaping in Time-Sensitive Networking}
\label{sec:background_cbs}
The CBS~\cite{standard, Bril} offers per-class bandwidth reservation and contributes to predictable latency by  ensuring bounded queuing delays on shared links, working in conjunction with other mechanisms like the time-aware shaper (TAS) for tight delay bounds in Ethernet. Each prioritized traffic class is assigned a dedicated queue, indexed by $p$. Each queue maintains a continuous-time credit, $C_p(t)$, and packets are transmitted only when $C_p(t) {\ge} 0$~\cite{Hassani}. Among queues that are both backlogged (i.e., have data to send) and meet this eligibility condition of having non-negative credit, CBS strictly prioritizes the highest priority queue for transmission, allowing preferential treatment and deterministic latency control for critical traffic. 
\subsection{Credit Dynamics: \texttt{IdleSlope} and \texttt{SendSlope}}
The credit $C_p(t)$ continuously  evolves over $t$ based on two CBS parameters and the queue's operational state:
\begin{enumerate}
    \item \textbf{\texttt{idleSlope} ($\alpha_p^+$):} The  rate (bits/s or bytes/s) at which credit \textit{accumulates} when  queue $p$ is eligible but not transmitting (due to being empty or blocked by higher priority). $\alpha_p^+$ represents the bandwidth share \emph{reserved} for  traffic class $p$ over time. 
    \item \textbf{\texttt{sendSlope} ($\alpha_p^-$):} The rate at which credit is \textit{consumed} while queue $p$ is transmitting. It is formally defined as: $ \alpha_p^- = \alpha_p^+ - R_{\text{link}}$, where $R_{\text{link}}$ is the physical link rate, making $\alpha_p^-$ negative, reflecting credit consumption during transmission. 
\end{enumerate}
If queue $p$ is idle (either no pending load and $C_p(t) \ge 0$, or $C_p(t) < 0$ and not transmitting, or pending load but blocked by other traffic), its credit increases. For any two  time instances $t_1, t_2$ and $t_2{>}t_1$ the credit updates as~\cite{Jingyue}:
    \begin{equation}
        C_p(t_2) = C_p(t_1) + \alpha_p^+ \cdot (t_2 - t_1) 
        \label{eq:credit_increase}
    \end{equation}
If queue $p$ is actively transmitting:
    \begin{equation}
        C_p(t_2) = C_p(t_1) - \alpha_p^- \cdot (t_2 - t_1) 
        \label{eq:credit_decrease}
    \end{equation}
The credit is reset to zero if the queue becomes empty and $C_p(t) \ge 0$ (Eq. 3.18 in~\cite{Jingyue}).
\subsection{Credit Limits: \texttt{hiCredit} and \texttt{loCredit}}
To prevent unbounded credit accumulation or deficit, CBS employs upper and lower bounds on $C_p(t)$:
\begin{enumerate}
    \item \textbf{\texttt{hiCredit} ($hi_p$):} The maximum positive credit value queue $p$ is allowed to accumulate. $hi_p$ limits the maximum burst size  that the queue $p$ can transmit immediately following an idle period. 
    \item \textbf{\texttt{loCredit} ($lo_p$):} The minimum credit value queue $p$ can reach, ensuring that even after transmitting a maximum-sized frame, the queue's credit deficit is bounded. This, in turn, bounds the maximum time required for the queue to recover positive credit solely through the accumulation via \texttt{idleSlope}. The credit  is  clamped within $[lo_p, hi_p]$. 
\end{enumerate}
Fig.~\ref{Fig_cred_evol} illustrates CBS with two classes.  
A frame starts only when its credit is non-negative and it has frames to transmit; \(p_1\) (high priority) is chosen over \(p_2\) if both qualify.  Credit accumulates  at the reserved rate \(\alpha_p^{+}\) when idle and is consumed 
at \(\alpha_p^{-}\) during transmission and are clamped at $[lo_{p1},  hi_{p1}]$ and $[lo_{p2}, hi_{p2}]$ respectively. Therefore, consumption and accumulation of credits prevent a single high priority queue from monopolizing the link indefinitely. 
\begin{figure}[t]
\centering
\includegraphics[scale=0.25]{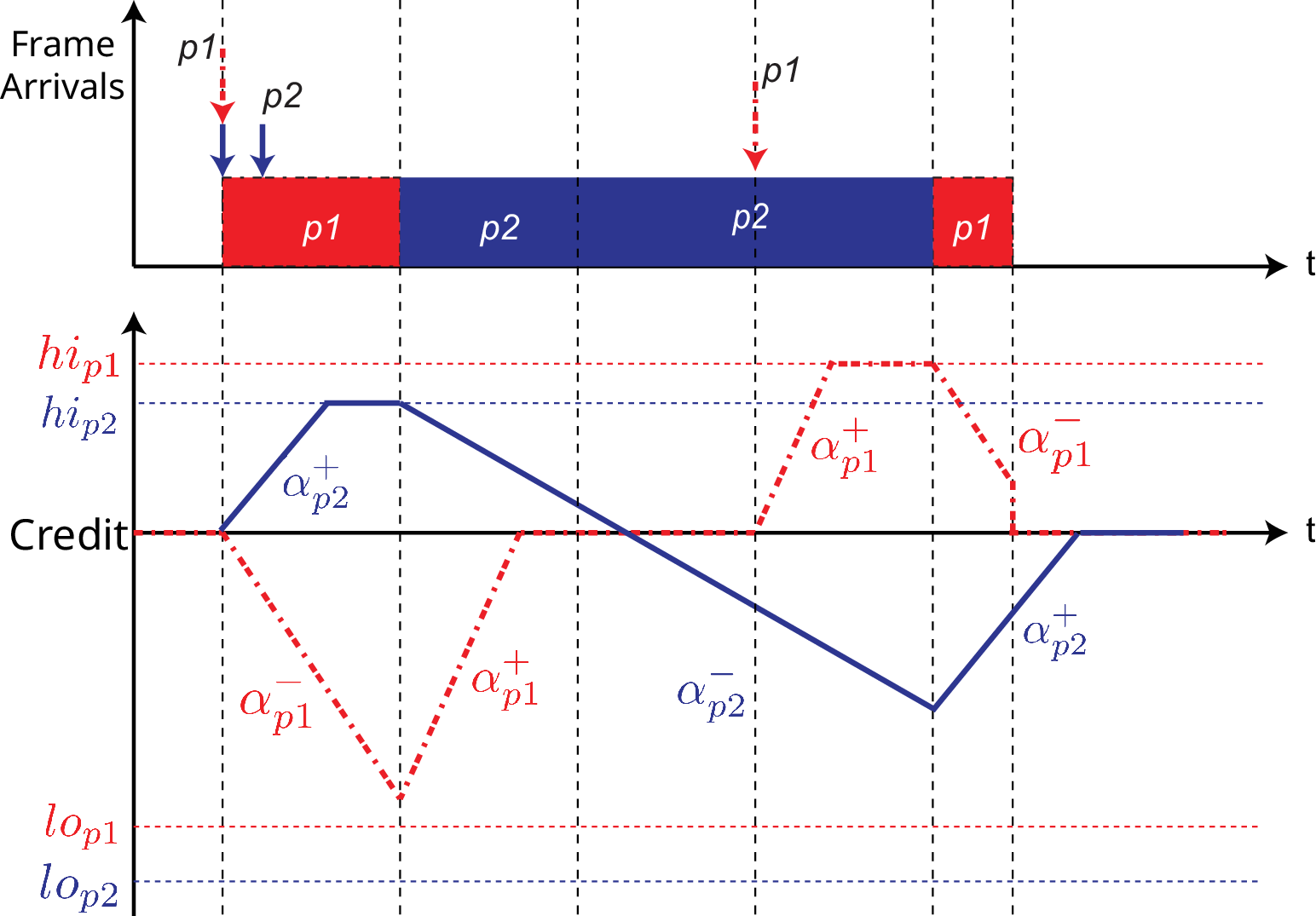}
\caption{CBS example with priorities $p_1$ (high) and $p_2$ (low)}
\label{Fig_cred_evol} 
\end{figure}
\subsection{Why Does Classical CBS Fail in a 5G NR Scheduler?} 
\label{subsec:cbs_challenges_wireless}
The key assumptions in IEEE 802.1Qav CBS are:
(i) \emph{continuous time} credit evolution;
(ii) a \emph{single, fixed link rate} $R_{\text{link}}$; and
(iii) shaping a single queue at a particular time instant that outputs frames without needing  complex resource allocation procedures.
On the other hand, 5G NR presents a fundamentally different operational context, invalidating these assumptions and posing significant challenges for directly applying CBS:
\subsubsection{Slot-quantized time} NR operates in discrete slots ($T_{\text{slot}}$), requiring discrete credit updates. 
\subsubsection{Grant-dependent rate} 5G service rates  are grant-dependent, unlike the fixed link rate in CBS. 
\subsubsection{Multi-user multiplexing} A wired port serves one frame at a time; NR can schedule multiple UEs/flows concurrently. 
\subsubsection{Feedback-coupled retransmissions} Automatic repeat request (ARQ) and Hybrid-ARQ (HARQ) can force a protocol data unit (PDU) to be transmitted multiple times; classical CBS is designed for reliable wires and  has no such mechanisms.

Consequently, these departures from the classical CBS--specifically,  discrete operational timeline, highly dynamic service rates tied to scheduling grants and channel quality, concurrent servicing of multiple users, and  the complexities introduced by wireless MAC/PHY feedback loops--invalidate key assumptions, making  direct porting  non-trivial. Effectively applying credit-based shaping within the 5G NR gNB mandates a rethinking of the mechanism, particularly requiring a discrete-time credit evolution logic accounting for the variable service provided in each transmission time interval (TTI). 
%
\label{sec_background}
\section{System Model}
\label{sec:system_model}
\begin{figure}[t]
\centering
\includegraphics[scale=0.7]{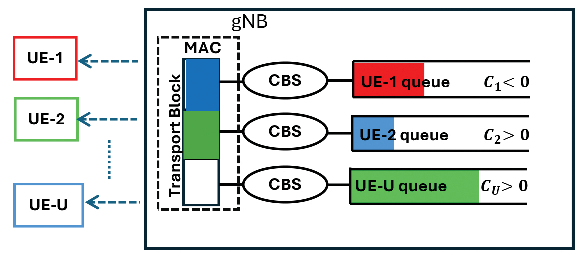}
\caption{Per-UE CBS gating in 5G NR gNB downlink}
\label{Fig_sys}
\end{figure}
We consider a 5G~NR downlink 
operating in frequency division duplex (FDD) with one gNB serving a finite set of UEs $\mathcal{U}=\{1,\dots,U\}$ as shown in Fig.\ref{Fig_sys}. Time is slotted: slots are indexed by $n\in\mathbb{Z}_{\ge 0}$ with duration $T_{\text{slot}}>0$.
\subsection{Time--frequency resources and decision variables}
In slot $n$ the scheduler can allocate up to $B_{\text{RB}}[n]\in\mathbb{Z}_{\ge 0}$ resource blocks (RBs) for \emph{data} on the DL shared channel. Let $x_u[n]\in\mathbb{Z}_{\ge 0}$ be the RBs granted to UE $u$ for \emph{new} transmissions in slot $n$, and let $y_u[n]\in\mathbb{Z}_{\ge 0}$ be the RBs devoted to HARQ (if any). The RB budget and control-channel limits impose
\begin{align}
&\sum_{u\in\mathcal{U}}\big(x_u[n]+y_u[n]\big)\ \le\ B_{\text{RB}}[n], \label{eq:rb-budget}\\
&\sum_{u\in\mathcal{U}}\mathbbm{1}\{x_u[n]>0\}\ \le\ K, \label{eq:k-cap}
\end{align}
where \eqref{eq:k-cap} captures the per-slot cap on \emph{new} grants due to PDCCH/CORESET capacity, with $K$ being the maximum number of new (non-HARQ) transmissions in a single slot. A new transmission for $u$ uses an MCS $m_u[n]\in\mathcal{M}$.
\subsection{Information pattern}
At the beginning of slot $n,$ the gNB observes
\[
\mathcal{F}_n {\!=\!} \Big( \{Q_u[n]\}_u,\ \{\text{CQI}_u[n]\}_u,\ \{\text{HARQ state}_u[n]\}_u,\ B_{\text{RB}}[n]\Big),
\]
where $Q_u[n]\in\mathbb{Z}_{\ge 0}$ is the DL backlog (bytes) known from buffer accounting, $\text{CQI}_u$ is the most recent channel-quality indicator (CQI), and the HARQ state specifies which processes await retransmission. The scheduler selects $(x_u[n],y_u[n],m_u[n])_{u\in\mathcal{U}}$ as $\mathcal{F}_n$-measurable decisions satisfying \eqref{eq:rb-budget}--\eqref{eq:k-cap}.
\subsection{Traffic, queues, and deadlines}
Each UE carries one QoS flow. 
Exogenous arrivals $A_u[n] \in \mathbb{Z}_{\ge 0}$ (bytes) enter the DL buffer of UE $u$; we allow the arrival process $\{A_u[n]\}$ to be stationary and ergodic. In our evaluation, we instantiate $\{A_u[n]\}$ as an ON/OFF traffic process to approximate the burstiness of aggregated TSN-like flows while keeping the load controllable.
The backlog recursion (in bytes) is:
$Q_u[n{+}1] {\;=\;} \max\!\big\{Q_u[n] - \mu_u[n],\,0\big\} \;+\; A_u[n],$
where $\mu_u[n]$ denotes the \emph{new-transmission service}, i.e., bytes removed from the radio link control (RLC) transmit buffer and packed into MAC PDUs for their first transmission. HARQ retransmissions are served from MAC HARQ process buffers without reducing $Q_u[n]$ and are never blocked by negative credit, which is consulted only when selecting new transmissions.
\subsection{PHY abstraction and service accounting}
Given an allocation of RBs $x_u[n]$ and an MCS $m_u[n]$, the physical layer determines a TBS $\mathrm{TBS}_u[n]$ (bytes) for that grant. Let $\mu_u[n]$ denote the bytes actually delivered by \emph{new} transmissions in slot $n$ (HARQ retransmissions are handled first). This value is determined by $\mu_u[n] = \min\{\mathrm{TBS}_u[n],\, Q_u[n]\}.$ If $\mathrm{TBS}_u[n] > Q_u[n]$, the excess is padding and does not change the queue state $Q_u[n]$; under CBS-PU this padding also does not incur credit debit (whereas CBS-DT would debit by $\mathrm{TBS}_u[n]$).
\subsection{HARQ process}
Each UE maintains $N_{\text{HARQ}}$ stop-and-wait processes. A NACK for a new transmission in slot $n$ enqueues a HARQ job that \emph{must} be scheduled as a retransmission in a future slot using $y_u[\cdot]$ resources until the transport block (TB) is successfully acknowledged or the process fails after exhausting a configured maximum number of retransmission attempts. Retransmissions consume RB budget but leave $Q_u[\cdot]$ unchanged (TB reside in HARQ memory). In every slot, retransmissions are scheduled \emph{prior} to new grants; thus, the residual budget may be smaller than $B_{\text{RB}}[n]$.
\subsection{Eligibility sets and scheduler}
Let $\mathcal{E}_{\text{MAC}}[n]$ be the set of UEs that are eligible for \emph{new} transmissions in slot $n$ from a pure MAC perspective:
\[
\mathcal{E}_{\text{MAC}}[n]{=}\big\{\, u{\in}\mathcal{U}:Q_u[n]{>}0,\ \text{a free HARQ process exists at }u \,\big\}.
\]
A credit-gating layer (Sections~\ref{sec_cbsstd}--\ref{sec_cbspu}) then filters $\mathcal{E}_{\text{MAC}}[n]$ to a set $\mathcal{E}[n]\subseteq \mathcal{E}_{\text{MAC}}[n]$; its internal dynamics are defined entirely in those sections and are not repeated here. A scheduler maps $\mathcal{E}[n]$ and $\mathcal{F}_n$ to $(x_u[n],m_u[n])$ under \eqref{eq:rb-budget}--\eqref{eq:k-cap}.
\subsection{Load and capacity notation}
Define the long-term mean offered rate $R_u\triangleq \lim_{N\to\infty}\frac{1}{N}\sum_{n=0}^{N-1}\mathbb{E}[A_u[n]]$ (bytes/slot), with aggregate offered load $\Lambda_{\text{DL}}=\sum_{u}R_u$. Let the \emph{measured} DL service rate under a given numerology/bandwidth/scheduler/channel be:
\begin{equation}
C_{\text{DL}} \;\triangleq\; \lim_{N\to\infty}\frac{1}{N}\sum_{n=0}^{N-1}\mathbb{E}\!\left[\sum_{u}\mu_u[n]\right]\!,
\end{equation}
and report the normalized load factor $\rho \triangleq \Lambda_{\text{DL}}/C_{\text{DL}}$. Let $C_{\text{res}}$ denote the long-term residual service rate available for \emph{new} transmissions after counting for HARQ retransmissions. For the rate reservation guarantees to be feasible, we assume an admission control policy is in place that ensures the system is not overprovisioned.  
\begin{equation}
    \sum_{u \in U} \Delta C_u \le C_{\text{res}}.
    \label{eq:admission_control}
\end{equation}
\paragraph*{Slot timeline}
At slot start: observe $\mathcal{F}_n$; \emph{(i)} allocate RBs to HARQ retransmissions ($y_u[n]$); \emph{(ii)} form $\mathcal{E}_{\text{MAC}}[n]$ and apply the credit gate to obtain $\mathcal{E}[n]$; \emph{(iii)} allocate RBs/MCS to new transmissions for $u\in\mathcal{E}[n]$ (respecting $K$); \emph{(iv)} data depart $Q_u$ according to $\mu_u[n]$; \emph{(v)} update credits; advance to $n{+}1$.
\section{Discrete-Time Credit-Based Shaping for 5G NR} 
\label{sec_cbsstd}
We now describe our proposed CBS-DT mechanism, which we have implemented at the gNB MAC layer to provide  priority differentiation among UEs. This mechanism adapts concepts from IEEE 802.1Qav to the slotted, resource-block-based structure of the 5G NR air interface.
\subsection{Per-UE Credit State and Dynamics}
 For each UE $u\in \mathcal{U}$ with assigned priority $p(u)$, we maintain a state variable $C_u[n]$, called \emph{credit},  representing the accumulated transmission allowance (in bytes) available at the \emph{beginning} of slot $n$. This credit value dictates the UE's eligibility for receiving new transmission grants. The credit evolves from $C_u[n]$ to $C_u[n+1]$ involving a credit  \emph{allowance} based on a configured reserved rate (analogous to \texttt{idleSlope}) and a credit \emph{debit}~$D_u[n]$ based on transmission activity in slot $n$. Specifically, we define the per-slot credit allowance, $\Delta C_u$ (bytes), derived as $\Delta C_u{\triangleq} \text{idleSlope}_{p(u)}{ \cdot} T_{\text{slot}}$. While continuous-time CBS also incorporates a debit mechanism based on its constant \texttt{sendSlope} (as detailed in Eq.~\eqref{eq:credit_decrease}), we propose a slot-dependent debit amount $D_u[n]$ (bytes).
 \begin{algorithm}[t!]
\caption{Per-slot credit update for UE $u$ at slot $n$}
\label{alg:credit_update_logic}
\begin{algorithmic}[1]
\Require $C_u[n]$ (current credit), $Q_u[n]$ (backlog in bytes), $\Delta C_u$ (allowance/slot),
         grant indicator $G_u[n]\!\in\!\{0,1\}$, grant size $\mathrm{TBS}_u[n]$ (bytes),
         bounds $lo_{u} < hi_{u}$, variant $\mathsf{VAR}\in\{\mathsf{DT},\mathsf{PU}\}$
\Ensure $C_u[n\!+\!1]$
\State \textbf{Accumulate / Idle / Reset}
\State $C^{\mathrm{pre}} \gets \mathrm{PreDebitUpdate}\big(C_u[n],\,\Delta C_u,\,Q_u[n]\big)$
\Comment{policy-defined; see note below}
\State \textbf{Compute debit}
\If{$G_u[n]=0$}
    \State $D_u[n] \gets 0$
\ElsIf{$\mathsf{VAR}=\mathsf{DT}$}
    \State $D_u[n] \gets \mathrm{TBS}_u[n]$ \Comment{debit by granted bytes}
\Else \Comment{$\mathsf{VAR}=\mathsf{PU}$}
    \State $D_u[n] \gets \min\{\mathrm{TBS}_u[n],\,Q_u[n]\}$ \Comment{debit by actually served bytes}
\EndIf
\State \textbf{Clamp}
\State $C_u[n\!+\!1] \gets \min\!\big\{\max\{\,C^{\mathrm{pre}}-D_u[n],\,lo_{u}\,\},\,hi_{u}\big\}$
\State \textbf{return} $C_u[n\!+\!1]$
\end{algorithmic}
\end{algorithm}
Let $X_u[n+1]$ represent the intermediate, unclamped credit  calculated at the \emph{end} of slot $n$ (before final clamping), based on the state $C_u[n]$ entering the slot and increments/debits occurring within slot $n$. The state evolves according to:
\begin{equation}
    X_u[n+1] = f(C_u[n], \Delta C_{u}, Q_{u}[n]) - D_u[n]
    \label{eq:unclamped_update}
\end{equation}
where $\Delta C_{u}$ is the per-slot credit allowance ($\text{idleSlope}_{p(u)} {\cdot }T_{\text{slot}}$), $D_u[n]$ is the debit term (derived later for each CBS variant), and the function $f(C, \Delta C^{\text{idle}}, Q^{\text{DL}})$ encapsulates the discrete-time adaptation of classical CBS credit accumulation, recovery, and reset rules \emph{prior to debiting}:
\begin{equation}
\small
    f(C, \Delta C, Q_{u}[n]) {\triangleq} 
            \begin{cases}
                \min(C {+ }\Delta C, 0) & \text{if } C{ <} 0 \\ 
                0 & \text{if } C {>} 0 \land Q_{u}[n]{=} \emptyset \\ 
                C {+ }\Delta C & \text{otherwise} 
            \end{cases}
    \label{eq:time_update_func}
\end{equation}
\normalsize
Here, $Q_{u}[n]> 0$ signifies a non-empty buffer  for UE $u$ and $D_u[n]$ is the debit term. The final credit state $C_u[n+1]$ is obtained by applying a clamping function $\Gamma(\cdot)$ that enforces the bounds $[lo_u, hi_u]$:
\small 
\begin{equation}
    C_u[n+1] = \Gamma(X_u[n+1]) \triangleq \min(\max(X_u[n+1], lo_u), hi_u)
    \label{eq:final_clamp_gamma_func}
\end{equation}
\normalsize
Note that the $\min(\cdot, lo_u)$ part of $\Gamma(\cdot)$ effectively applies to the output of $f(\cdot)$ before the debit $D_u[n]$ is subtracted in~\eqref{eq:unclamped_update} if we strictly follow the multi-step process; however, for analysis,~\eqref{eq:unclamped_update}--\eqref{eq:final_clamp_gamma_func} capture the overall state transition.
%
%
For the CBS-DT adaptation, the debit $D_u[n]$ is coupled to resource granted by the MAC scheduler in slot $n$. If a new transmission grants yields $\text{TBS}_u[n]{>}0$ bytes, the debit equals the full grant size: 
\begin{equation}
D_u[n] = \text{TBS}_u[n] \quad (\text{CBS-DT})
\label{eq:debit_std_tbs_final}
\end{equation}
If no new grant is issued, $D_u[n]{=}0$. Since the  $\text{TBS}_u[n]$  varies with channel quality and resource allocation, $D_u[n]$ is inherently a stochastic quantity unlike in CBS.  
\subsection{CBS-DT Stability Conditions}
\label{subsec:stochastic}
The use of a \emph{grant-dependent} debit, $D_u[n]$~\eqref{eq:debit_std_tbs_final}, introduces stochasticity tied to CQI and resource allocation, departing significantly from the CBS. Nevertheless, the fundamental stability conditions can be established, which are: (i) $C_u[n]$ must remain bounded $\forall$ $n$; (ii) queues must recover from credit deficits. Credit boundedness is directly enforced  by the clamping function $\Gamma(\cdot)$~\eqref{eq:final_clamp_gamma_func} where $C_u[n] \in [lo_u, hi_u]$ $\forall$ $n \ge 0$. Deficit recovery (reaching $C_u[n]{ \ge }0$ from any $C_u[k]{<}0$, including the lower bound $lo_u$) is guaranteed because the per-slot allowance $\Delta C_{u}{ > }0$ provides a constant positive increment towards the non-negative region whenever no debit occurs~\eqref{eq:time_update_func}.  
\subsection{Eligibility Criterion and MAC Interaction}
\label{subsec:eligibility}
The proposed CBS-DT mechanism shapes traffic by controlling the eligibility of each UE based on accumulated credits.  The overall eligibility set $\mathcal{E}[n]$ at slot $n$ is determined by obtaining the underlying MAC scheduler eligibility set  $\mathcal{E}_{\text{MAC}}[n]$: 
\begin{equation}
\mathcal{E}_{\text{MAC}}[n] = \left\{
\begin{aligned}
& u \in \{1, \dots, U\},\ Q_u[n] > 0, \\
& \text{HARQ}_u\ \text{available},\ u \notin \text{ReTxSet}[n]
\end{aligned}
\right\}
\end{equation}
Here, $\text{HARQ}_u\text{ available}$ indicates user $u$ has free HARQ process and $u \notin \text{ReTxSet}[n]$ ensures that user $u$ is not already scheduled for HARQ retransmission in $n$. The CBS-DT further restricts eligibility by checking the credit condition, $\mathcal{E}[n]{=}{u \in \mathcal{E}_{\text{MAC}}[n]} \large{\mid} C_u[n]\geq0$.
Thus, CBS-DT serves as a credit-based gate layered atop the existing MAC layer.  
\section{Credit-Based Shaping with Partial Usage}
\label{sec_cbspu}
A key challenge arises with the CBS-DT adaptation due to the practicalities of 5G NR scheduling. While the gNB knows the precise downlink queue size, the selected $\text{TBS}[n]$ must conform to the quantized values mandated by {3GPP}~\cite{ts38214}.  Consequently whenever the granted \(\text{TBS}_u[n]\) exceeds $Q_u[n]$, padding is added. CBS-DT, by debiting the \emph{entire} grant $\text{TBS}_u[n]$~(\ref{eq:debit_std_tbs_final}), therefore over-penalizes UE credit. This drives the credit deeper into deficit postponing the next eligibility event, delaying subsequent transmission eligibility particularly affecting UEs with bursty arrival traffic patterns.  We therefore introduce \emph{Credit-Based Shaping with Partial Usage (CBS-PU)}, retaining the CBS-DT structure outlined in Section~\ref{sec_cbsstd}. The sole modification lies in the definition of the debit term $D_u[n]$ used in~(\ref{eq:unclamped_update}).
\subsection{CBS-PU Debit Mechanism}
\label{subsec:pu_debit}
The defining characteristic of CBS-PU is that the credit decrement $D_u[n]$ is proportional to the \emph{actual data transmitted} within the granted resources, rather than the full granted $\text{TBS}_u[n]$.
The actual bytes for transmission can be calculated as:
\begin{equation}
    B_{u}^{\text{actual}}[n] = \min(\text{TBS}_u[n], Q_{u}[n])
    \label{eq:actual_bytes_pu_rev}
\end{equation}
The debit term $D_u[n]$ in~\eqref{eq:unclamped_update} under CBS-PU is:
\begin{equation}
    D_u[n] = B_{u}^{\text{actual}}[n] \quad (\text{CBS-PU})
    \label{eq:debit_pu_actual_final_rev}
\end{equation}
If no new grant is issued ($\text{TBS}_u[n]=0$), then $B_{u}^{\text{actual}}[n]=0$ and thus $D_u[n]=0$. This ensures credit is consumed only according to the utilized portion of the grant.
\subsection{Theoretical Properties}
\label{subsec:pu_properties}
CBS-PU retains the stability conditions established for the  CBS-DT adaptation while offering improved performance characteristics. Under the admission control assumption in \eqref{eq:admission_control}, the gNB has sufficient long-term resources to satisfy all configured rate reservations. Consequently, for a persistently backlogged queue, the fraction of time spent at the upper credit limit is negligible. This allows the rate preservation property to be formally established: 
\begin{theorem}
\label{thm:pu_properties}
The proposed CBS-PU mechanism ensures:
\begin{enumerate}
    \item Credit Dominance: $C_u^{\text{PU}}[n] \ge C_u^{\text{DT}}[n]$ for all $n \ge 0$, given identical initial conditions and system inputs.
    \item Rate Preservation: For persistently backlogged queues, the long-term average service rate converges to configured per-slot credit allowance: $\lim_{N \to \infty} \frac{1}{N} \sum_{n=0}^{N-1} B_u^{\text{actual}}[n]{=}\Delta C_u$.
    \item Delay Bound Preservation: The worst-case delay under CBS-PU is never greater than under CBS-DT: $W_u^{\text{PU}}~{\le}~W_u^{\text{DT}}$.
\end{enumerate}
\end{theorem}
\begin{proof}[Proof sketch]
(a) Follows by induction on $n$.  Assume $C_u^{\text{PU}}[k]{ \ge }C_u^{\text{DT}}[k]$. Since $f(\cdot)$ and $\Gamma(\cdot)$ are monotonic non-decreasing with respect to credit input, and $D_u^{\text{PU}}[k] = B_{u}^{\text{actual}}[k] \le \text{TBS}_u[k] = D_u^{\text{DT}}[k]$, applying Eqs.~\eqref{eq:unclamped_update}-\eqref{eq:final_clamp_gamma_func} shows $C_u^{\text{PU}}[k+1] \ge C_u^{\text{DT}}[k+1]$.
(b) Let $\alpha_u^+ = \text{idleSlope}_{p(u)}$ be the configured reserved rate. Credit $C_u[n]$ is bounded in $[lo_u, hi_u]$. For a persistently backlogged queue over a long interval $N$, the total credit change $\sum_{n=0}^{N-1} (\Delta C_u - D_u[n])$ must remain bounded. Thus, $\frac{1}{N}\sum D_u[n] \to \Delta C_u$. Since $D_u[n] = B_u^{\text{actual}}[n]$ and $\Delta C_u = \alpha_u^+ \cdot T_{\text{slot}}$, the long-term average actual transmitted rate $\frac{1}{N T_{\text{slot}}} \sum B_u^{\text{actual}}[n]$ converges to $\alpha_u^+$. (c) Follows directly from credit dominance. Since $C_u^{\text{PU}}[n] \ge C_u^{\text{DT}}[n]$, a queue cannot become eligible later under CBS-PU than under CBS-DT. With the same MAC scheduling policy for eligible UEs, the departure time under PU must be less than or equal to that under CBS-DT for any given arrival, hence $W_u^{\text{PU}} {\le }W_u^{\text{DT}}$.
\end{proof}
Algorithm~\ref{alg:credit_update_logic} summarizes the slot-level credit recursion: at each slot we (i) apply the pre-debit update \(f(\cdot)\) in Eq.~\eqref{eq:time_update_func} (accumulation, recovery, and positive-credit reset), (ii) subtract a debit equal to the granted TBS (CBS-DT) or the actually delivered bytes (CBS-PU), and (iii) clamp the result to \([lo_u,hi_u]\). Only \emph{new} (non-HARQ) grants incur a debit. The primary advantage of CBS-PU arises from property (a). By maintaining higher (or equal) credit levels compared to CBS-DT, CBS-PU facilitates faster recovery from credit deficits. This translates to reduced time spent ineligible, and enables higher resource utilization efficiency, as validated in Section~\ref{subsec:utilization}.
\section{Computational Complexity and Scalability}\label{sec:complexity}
The practical adoption of the proposed CBS mechanisms hinges on their feasibility within a gNB’s real-time MAC budget. This section analyzes runtime and memory complexity under two, naive and event-driven, implementations.
\subsection{Assumptions used in analysis}
\label{subsec:assumptions}
For every traffic class, the per-slot allowance is strictly positive: $\Delta C_u = \texttt{idleSlope}_{p(u)} T_{\text{slot}} > 0$.  Credit is debited only on a \emph{new} transmission grant (no debit on HARQ retransmissions), consistent with Section~\ref{sec:system_model}. The MAC selector over the \emph{eligible} set is RR with at most $K$ new grants per slot; numerology is fixed with $T_{\text{slot}}=1$\,ms.
\subsection{Naive Implementation}
A straightforward implementation updates credit and checks eligibility for \emph{all} $U$ UEs every slot.
\begin{proposition}[Worst-case per-slot cost]
A naive implementation of CBS-DT or CBS-PU has per-slot cost $T_{\text{slot}}^{\text{CPU}} = \Theta\big(U + G_n\big),$
where $G_n$ is the number of \emph{new} grants issued in slot $n$.
\end{proposition}
\noindent\begin{proof}[Justification]
Per-slot credit evolution and the predicate $C_u[n]{\ge}\,0$ cost $\Theta(U)$; RR assignment costs $\Theta(G_n)$. The debit for DT ($\texttt{TBS}_u[n]$) or PU ($\min\{\texttt{TBS}_u[n],Q_u[n]\}$) is $O(1)$.
\end{proof}
\subsection{Event-Driven Implementation for Amortized Efficiency}
The naive loop wastes \emph{work}, which is per-TTI MAC bookkeeping (read/update \(C_u\), eligibility tests, queue/heap touches), on UEs in deficit (\(C_u<0\)), whose dynamics are deterministic while no new grant occurs. We replace per-slot updates with wake-ups: for any \(u\) with{ \(C_u[n]<0\)}, compute {\(w_u = n + \left\lceil \frac{-C_u[n]}{\Delta C_u} \right\rceil\)},
insert \((u,w_u)\) into a min-heap keyed by wake-up slot, and skip updates for \(u\) until \(n=w_u\).
\begin{lemma}[Safety of skipping]
Skipping applies only when $C_u{<}\,0$. While in deficit and without new grants,~(\ref{eq:time_update_func}) gives $C_{n+1}{=}\min(C{+}\Delta C, 0)$, so $C$ reaches $0$ exactly at $w_u$. For $C_u{>}\,0$, the \emph{positive-credit reset on empty queue} ($C{>}\,0 \wedge Q{=}\,\varnothing \Rightarrow \text{$C$ is reset to zero}$~(\ref{eq:time_update_func})) is handled as a \emph{local} event during post-grant bookkeeping for that UE; without a global scan.
\end{lemma}
\begin{theorem}[Amortized complexity]
Over $N$ slots, the total work of the event-driven design is
\begin{equation}
    T_{1{:}N} = O(N + G) + O\big((A{+}G)\log U\big),
\end{equation}
where $A$ is the number of \emph{queue-becomes-non-empty} events and $G=\sum_n G_n$ is the total number of \emph{new-grant} events.
\end{theorem}
\begin{proof}[Proof sketch]
Each slot performs an $O(1)$ heap-min peek, giving $O(N)$. RR grants cost $O(G)$ in total with the circular eligible list. Heap insertions arise when (i) a grant drives a UE into deficit or (ii) an arrival finds a UE indebted; thus at most $A{+}G$ insertions. Wake-ups are at most the number of insertions ($A{+}G$), while each heap operation costs $O(\log U)$, yielding the second term, $O\big((A{+}G)\log U\big)$.
\end{proof}
This event-driven implementation shifts cost from the UE population $U$ to \emph{event counts} $(A,G)$. In typical loads where many UEs spend time in deficit, $(A{+}G){\ll}\,U\cdot N$, and the scheduler’s cost tracks traffic, not population. Algorithm~\ref{alg:event_driven_cbs} gives the event-driven per-slot driver for the credit gate.
\begin{algorithm}[t]
\caption{Event-Driven CBS: Per-Slot Logic}
\label{alg:event_driven_cbs}
\begin{algorithmic}[1]
\Function{OnSlotStart}{$n$, eligible\_list, wakeup\_heap}
    \State \Comment{Process wake-ups}
    \While{wakeup\_heap.notEmpty() \textbf{and} wakeup\_heap.minKey() == $n$}
        \State $u \gets$ wakeup\_heap.extractMin()
        \State $u.\text{credit} \gets 0$
        \If{$u.\text{queueNotEmpty}()$ \textbf{and} $u.\text{hasFreeHARQ}()$}
            \State eligible\_list.add($u$)
        \EndIf
    \EndWhile
    \State \Comment{Serve HARQ first; updates RB budget and \texttt{num\_grants}}
    \State ScheduleHARQ()
    \State \Comment{Issue up to \texttt{num\_grants} RR allocations (bounded by RBGs)}
    \State granted $\gets$ RoundRobin(eligible\_list, num\_grants)
    \State \Comment{Post-grant bookkeeping for granted UEs}
    \For{$u$ \textbf{in} granted}
        \State debit $\gets$ CalculateDebit($u$) \Comment{$O(1)$ for DT or PU}
        \State $u.\text{credit} \gets u.\text{credit} - \text{debit}$
        \State $u.\text{credit} \gets \max(u.\text{credit},\,u.\text{lo})$ \Comment{clamp to $lo_u$ for wake calc}
        \If{$u.\text{queueIsNowEmpty}()$}
            \If{$u.\text{credit} > 0$}
                \State $u.\text{credit} \gets 0$ \Comment{positive-credit reset}
            \EndIf
            \State eligible\_list.remove($u$)
        \ElsIf{$u.\text{credit} < 0$}
            \State eligible\_list.remove($u$)
            \State wake $\gets n + \left\lceil \frac{-u.\text{credit}}{u.\text{delta\_c}} \right\rceil$
            \State wakeup\_heap.insert($u$, key=wake)
        \EndIf
    \EndFor
\EndFunction
\end{algorithmic}
\end{algorithm}
\subsection{Recovery Time and Discussions}
With a credit lower bound $lo_u{<}\,0$ and a per-slot allowance $\Delta C_u{>}\,0$, the  recovery time from a worst-case deficit, assuming no grants are issued,  is a deterministic value given by:
   $ T^{\text{rec}}_{\text{max}} = \Big\lceil -\,lo_u / \Delta C_u \Big\rceil \cdot T_{\text{slot}},$
which is \textit{tight}. An event-driven implementation exploits the deterministic negative-credit drift to yield amortized cost that tracks traffic events rather than population. By Theorem~1 (credit dominance \& no-worse delay), CBS-PU reduces deficit episodes, decreasing heap events and improving practical runtime over CBS-DT.
\subsection{Worst-Case Delay and Pacing Guarantees}
The deterministic nature of slot-native CBS allows derivation of hard bounds on waiting time and inter-grant separation. \paragraph*{Assumption (Bounded eligible burst)}
There exists a constant $E_{\max}\!\ge\!1$ such that, in any slot $n$, the number of UEs (excluding $u$) that are eligible (credit $\ge 0$) and still awaiting a \emph{new} grant is at most $E_{\max}$, enforced by admission control.
\begin{lemma}[Absolute waiting bounds] 
\label{lem:absolute-wait}
For UE $u$ with per-slot allowance {$\Delta C_u \triangleq \texttt{idleSlope}_{p(u)} T_{\mathrm{slot}} {> }0$}, credit clamp
$[lo_{u},hi_{u}]$ eligibility gate $C_u{\ge} 0$, and at most $K$ new grants/slot \textup{(cf.\ (4))} under RR over the eligible set, the following hold.

\textit{(A) Time to eligibility from any deficit.}
If a packet arrives at slot $n_0$ when $C_u[n_0]<0$, then
\begin{equation}
W^{\mathrm{elig}}_u \le \Big\lceil \frac{-C_u[n_0]}{\Delta C_u} \Big\rceil T_{\mathrm{slot}}
\le \Big\lceil \frac{-\ell^{u}_{\mathrm{lo}}}{\Delta C_u} \Big\rceil T_{\mathrm{slot}}.
\end{equation}
%
\textit{(B) Time to first grant once eligible.}
Let $n_{\mathrm{elig}}$ be the first slot with $C_u[n_{\mathrm{elig}}]\ge 0$. Let $E_{\text{max}}$ be a constant that bounds the number of eligible backlogged UEs over the horizon.  The additional wait \emph{after} eligibility is then bounded by,  $W^{\mathrm{queue}\mid\mathrm{elig}}_u \le \Big\lceil \frac{E_{\max}}{K} \Big\rceil T_{\mathrm{slot}}.$
Hence the total time to the first grant is
\begin{equation}
W^{\mathrm{svc}}_u \le W^{\mathrm{elig}}_u + \Big\lceil \frac{E_{\max}}{K} \Big\rceil T_{\mathrm{slot}},
\end{equation}
\textbf{(C) Inter-grant separation (pacing).}
If a new grant at slot $n^\star$ debits $D_u[n^\star]\le D_u^{\max}$, the post-grant credit satisfies
$C_u[n^\star{+}] \ge \max\{lo_{u},-D_u^{\max}\}$. Thus the largest possible post-grant deficit magnitude is
\[
\delta^{\max}_u \triangleq -\max\{lo_{u},-D_u^{\max}\} \;=\; \min\{-lo_{u},\,D_u^{\max}\}.
\]
Therefore the time until the UE becomes eligible again is bounded by
\begin{equation}
    W^{\mathrm{re-elig}}_u \;\le\; \Big\lceil \frac{\delta^{\max}_u}{\Delta C_u} \Big\rceil T_{\mathrm{slot}}
    \;=\; \Big\lceil \frac{\min\{-lo_{u},\, D_u^{\max}\}}{\Delta C_u} \Big\rceil T_{\mathrm{slot}}.
\end{equation}
Equivalently, in two regimes,
\[
W^{\mathrm{re-elig}}_u \;\le\;
\begin{cases}
\Big\lceil \dfrac{D_u^{\max}}{\Delta C_u} \Big\rceil T_{\mathrm{slot}}, & \text{if } D_u^{\max}\le -lo_{u},\\[0.8em]
\Big\lceil \dfrac{-lo_{u}}{\Delta C_u} \Big\rceil T_{\mathrm{slot}}, & \text{if } D_u^{\max}\ge -lo_{u}.
\end{cases}
\]
The total time until the next grant, $W^{\mathrm{cycle}}_u$, including subsequent queuing, is then
\begin{equation}
    W^{\mathrm{cycle}}_u \le W^{\mathrm{re-elig}}_u + \Big\lceil \frac{E_{\max}}{K} \Big\rceil T_{\mathrm{slot}}.
\end{equation}
\end{lemma}
\section{ Performance Evaluation}\label{sec_results}
\begin{figure}[b!]
    \centering 
    \begin{subfigure}[b]{0.48\columnwidth}
        \includegraphics[scale=0.42]{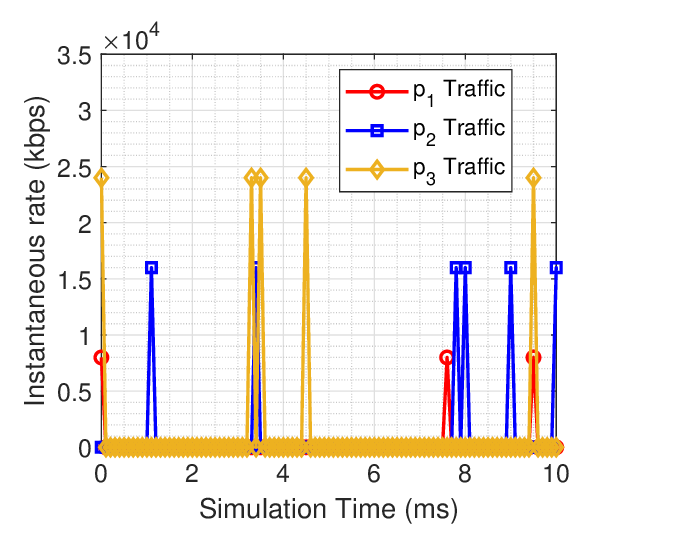} 
        \caption{Traffic-arrival pattern} 
        \label{fig:grid_tl} 
    \end{subfigure}
    \begin{subfigure}[b]{0.48\columnwidth}
        \centering
        \includegraphics[scale=0.42]{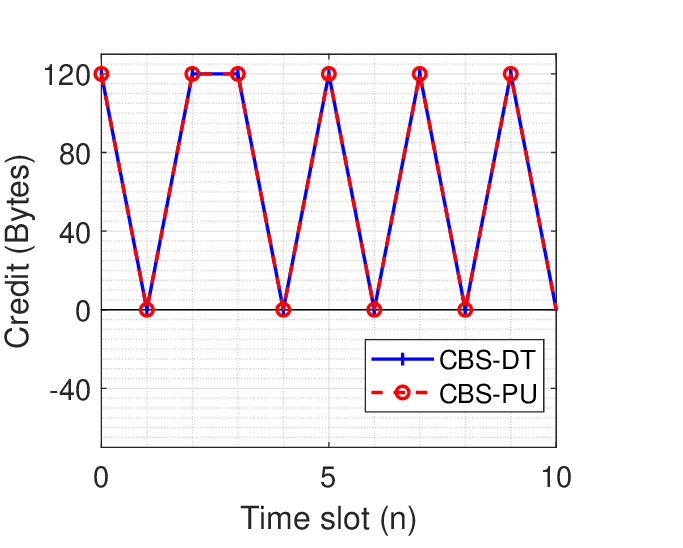} 
        \caption{$p_1$ - credits coincide.} 
        \label{fig:grid_tr}
    \end{subfigure}
    \begin{subfigure}[b]{0.48\columnwidth} 
        \includegraphics[scale=0.42]{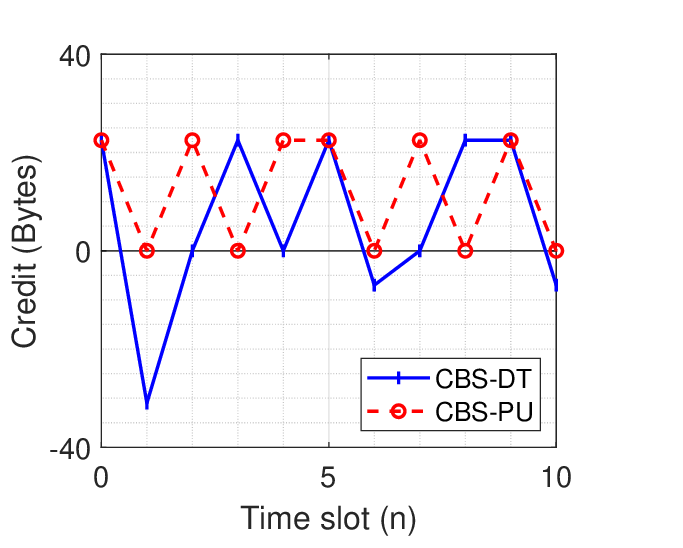}
        \caption{$p_2$ - PU stays higher.}
        \label{fig:grid_bl}
    \end{subfigure}
    \begin{subfigure}[b]{0.48\columnwidth} 
        \centering
        \includegraphics[scale=0.42]{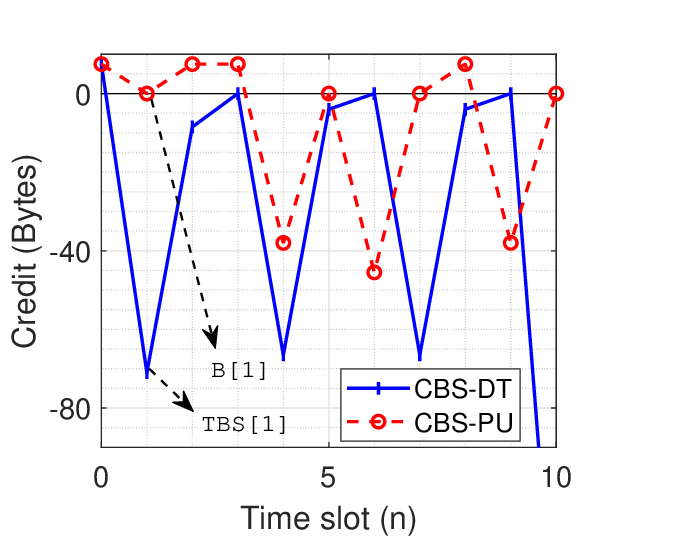} 
        \caption{$p_3$ - PU avoids over-debit.} 
        \label{fig:grid_br}
    \end{subfigure}
    \caption{Overview of traffic arrivals and credit evolution.} 
    \label{fig:grid_combined}
\end{figure}
To evaluate the performance of the proposed CBS-DT and CBS-PU mechanisms, we built a simulation framework upon the MATLAB 5G Toolbox. We simulate a single-cell FDD NR downlink with a 15~kHz subcarrier spacing (SCS), using the 3GPP Urban Macro (UMa) propagation model (TR 38.901); CQI reports are mapped to MCS according to 3GPP TS 38.214. 
Six static UEs are evenly split across three priority classes: $p_1$ (high), $p_2$ (medium), and $p_3$ (low). Each UE generates a packet stream with nominal rate of $\lambda = 450$~pkt/s and a fixed payload size $S_u \in \{80, 160, 240\}$~bytes for $p_1$, $p_2$, and $p_3$, respectively. The aggregate offered rate is defined as $\Lambda_{\text{DL}} = \sum_{u} \lambda S_u,$ and we report the normalized load $\rho = \Lambda_{\text{DL}} / C_{\text{DL}}$; we study both nominal load ($\rho=1.0$) and a controlled overload ($\rho=4.0$) to stress-test the scheduler's priority protection capabilities.
Class reservations (idleSlope shares) are set to \{75\%, 20\%, 5\%\} for priorities \{$p_1, p_2, p_3$\}, respectively. Per-UE credit clamps $[lo_{u}, hi_{u}]$ are derived following the Ethernet CBS standard~\cite{standard}. The credit gate renders a UE eligible in slot $n$ iff it is backlogged, has a free HARQ process, and its credit is non-negative. HARQ retransmissions consume radio resources but do not reduce the downlink queue $Q_u$; therefore, only the bytes from a first transmission debit the credit.
The CBS-PU mechanism relies on information readily available at the gNB. At the start of each slot, the gNB knows the queue size $Q_u[n]$, the chosen TBS, and the HARQ state. A new transmission service removes $\min\{\text{TBS}_u[n], Q_u[n]\}$ bytes from the buffer. Any padding added to the TB does not change $Q_u$. 
\subsection{Credit Evolution Analysis}
\label{subsec:credit_evo}
Fig.~\ref{fig:grid_combined} depicts representative DL traffic arrivals (Fig.~\ref{fig:grid_tl}) with credit ($C_u[n]$) trajectories of one UE per priority class under nominal load for  CBS-DT and CBS-PU.  High-priority credits, as shown in Fig.~\ref{fig:grid_tr}, oscillate between 0 and the configured high limit ($hi_u{\approx}\!120$ B) for both variants. The grants are fully utilized, preventing padding. Consequently, the debit terms are identical, resulting in coinciding credit trajectories. In such no-padding conditions, CBS-PU offers no additional benefit over CBS-DT, consistent with Theorem~\ref{thm:pu_properties}(a). Middle-priority UEs (Fig.~\ref{fig:grid_bl}) show more pronounced credit fluctuations. Under CBS-DT, credits dip moderately into negative territory following transmissions. With CBS-PU, these negative excursions are shallower and recovery is quicker, as the debit accurately reflects smaller actual transmissions ($B_{u}^{\text{actual}}[n]$) rather than the full $\text{TBS}_u[n]$. The largest contrast appears in  Fig.~\ref{fig:grid_br}.  CBS-DT frequently hits the lower credit values,  requiring lengthy recovery, while CBS-PU significantly mitigates these deficits by debiting only the actual usage. 
\begin{figure}[t] 
\captionsetup{justification=centering}
    \begin{subfigure}[b]{0.5\columnwidth}
        \centering
        \includegraphics[scale=0.48]{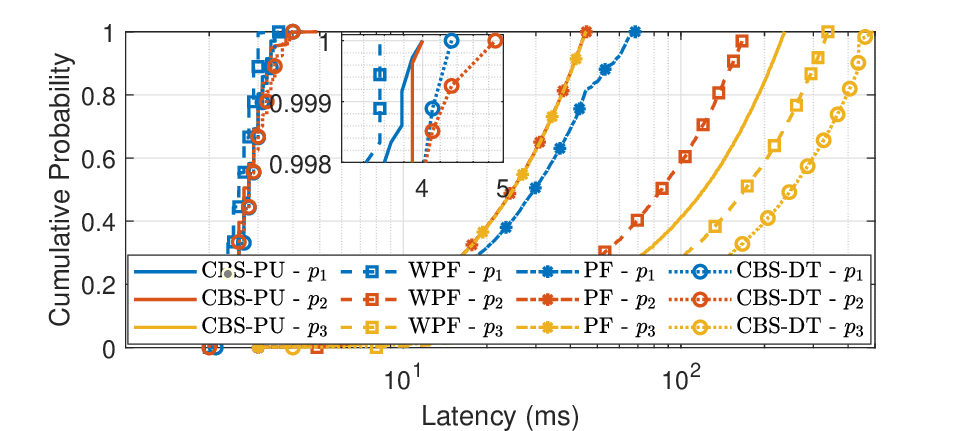} 
        \caption{Higher load ($\rho=4$)} 
        \label{fig:latency_cdf_rho100} 
    \end{subfigure}
    \begin{subfigure}[b]{0.5\columnwidth}
        \centering
        \includegraphics[scale=0.48]{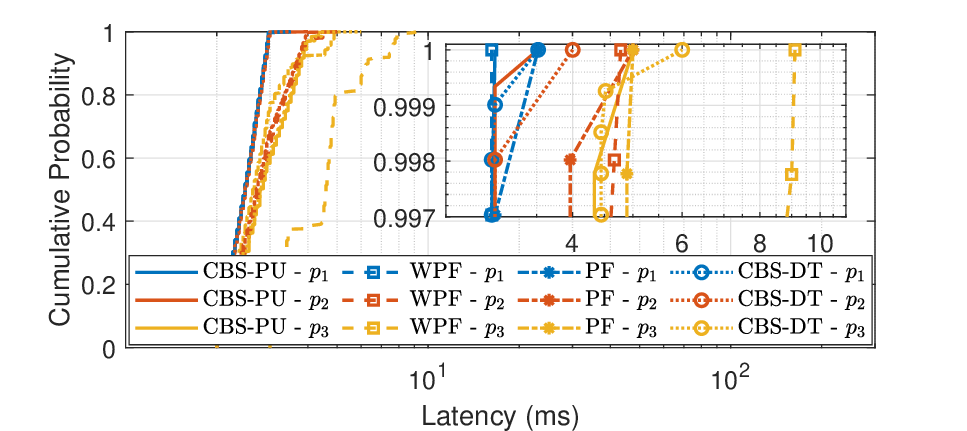} 
        \caption{Nominal load ($\rho=1.0$)} 
        \label{fig:latency_cdf_rho400}
    \end{subfigure}
    \caption{Empirical CDFs of downlink packet latency} 
    \label{fig:latency_cdf_combined} 
\end{figure}
\subsection{Latency and QoS Performance}\label{subsec:latency}
Fig.~\ref{fig:latency_cdf_combined} shows downlink latency cumulative distribution functions (CDFs) for four schedulers: CBS-DT and CBS-PU (each layered over an RR selector), PF, and WPF. We compare a nominal load ($\rho=1.0$) to a heavy load ($\rho=4.0$) chosen to stress priority enforcement. To accentuate scheduler behavior, high-priority UEs  operate under poorer channel conditions. 
\subsubsection{Heavy load ($\rho=4.0$)}
PF is class-agnostic: it schedules $i=\arg\max s_i[n]$ with $s_i[n]=r_i[n]/\bar R_i[n]$, where $r_i[n]$ is the instantaneous
rate estimate implied by the current MCS and RB budget, and $\bar R_i[n]$ is the moving average of past served rates. PF exhibits a distributional priority inversion: over a nontrivial range of $n$, $\mathrm{CDF}^{\mathrm{PF}}_{p_1}(n)<\mathrm{CDF}^{\mathrm{PF}}_{p_2}(n)$ and $\mathrm{CDF}^{\mathrm{PF}}_{p_1}(n)<\mathrm{CDF}^{\mathrm{PF}}_{p_3}(n)$. UEs with smaller $\bar R_i[n]$ or transiently larger $r_i[n]$ are favored, \emph{regardless of class}; in our setup this often benefits $p_2/p_3$, producing the observed inversion. Class weighting (proportional to idleSlope) shifts the $p_1$ curve left relative to PF, restoring order for $p_1$ while pushing $p_2$--and especially $p_3$--rightward, reflecting the higher long-term share assigned to $p_1$. Both CBS-PU and CBS-DT preserve class order and markedly tighten tails versus PF/WPF. For $p_2$, CBS-DT stochastically dominates PF and WPF. For $p_3$, CBS-DT shows the heaviest tail in the panel, whereas CBS-PU yields visibly better $p_3$ latency. The gap stems from \emph{partial-usage} debiting: CBS-PU charges delivered bytes rather than a full TB, yielding a smaller post-grant debit and a shorter re-eligibility gap (cf.\ Lemma~\ref{lem:absolute-wait}(C)); full-TB debits in CBS-DT drive deeper deficits and longer recovery.
\subsubsection{Nominal load ($\rho=1.0$)}
At $\rho=1.0$, queues are stable and $\bar R_i[n]$ tracks long-term rates closely, so PF’s catch-up effect is weak. PF preserves class order but retains heavier tails than CBS for lower-priority classes. WPF moves the $p_1$ CDF left of PF while pushing $p_3$ right, reflecting a heavier tail for the lowest-priority class and indicating starvation. Both CBS variants preserve class order and produce tighter tails; $p_2$ remains close to $p_1$, and $p_3$ is limited but not starved.
\subsubsection{Determinism and the role of RR} CBS provides slot-level determinism: time to regain eligibility after any deficit is bounded (Lemma~\ref{lem:absolute-wait}(A)), and, once eligible, RR over the eligible set with at most $K$ new grants per slot gives a deterministic access bound $\lceil E_{\max}/K\rceil T_{\mathrm{slot}}$ (Lemma~\ref{lem:absolute-wait}(B)). Combined with the re-eligibility bound (Lemma~\ref{lem:absolute-wait}(C)), this yields a per-cycle latency bound. These guarantees require a deterministic selector (RR) for the eligible set; if PF/WPF selects among eligible UEs, the $\lceil E_{\max}/K\rceil$ bound does not apply and no comparable worst-case delay guarantee is available.
\subsection{Resource Utilization Efficiency}
\label{subsec:utilization}
We quantify the efficiency of translating granted resources into actual data transmission via $\eta_u {=} (\sum B_{u}^{\text{actual}} / \sum \text{TBS}_u){\times}100\%$. Fig.~\ref{fig:utilization_low_load} presents $\eta_u$ evaluated at reduced load ($\rho=0.2$~Fig.~\ref{fig:rho20}), chosen specifically because partial grant fills ($B_{u}^{\text{actual}}[n] {< }\text{TBS}_u[n]$),  exposing utilization inefficiencies, that are more prevalent  at $\rho=0.2$ than at $\rho = 0.4$, where buffers are often saturated as shown in~Fig.\ref{fig:rho40}. CBS-PU minimizes \emph{credit waste}, maximizing efficiency ($>\!98\%$) by debiting only $B_{u}^{\text{actual}}[n]$,  avoiding penalizing queues for unused grant portions.
CBS-DT suffers efficiency loss from TBS debiting. Its  $\text{TBS}_u[n]$ debit creates \emph{credit waste} when $B_{u}^{\text{actual}}[n] < \text{TBS}_u[n]$. This not only reduces the $\eta_u$ ratio directly but also prolongs credit recovery, hence increasing the latency as observed in Section~\ref{subsec:latency}. The resulting efficiency ($\approx$95--98\%) reflects this penalty.
RR's  policy, despite serving any backlogged UE without shaping, results in frequent resource mismatches where $B_{u}^{\text{actual}}[n] \ll \text{TBS}_u[n]$, yielding the lowest efficiency ($\approx$80--85\% for $p3$).
\begin{figure}[t] 
    \begin{subfigure}[b]{0.48\columnwidth}
        \includegraphics[width=\linewidth]{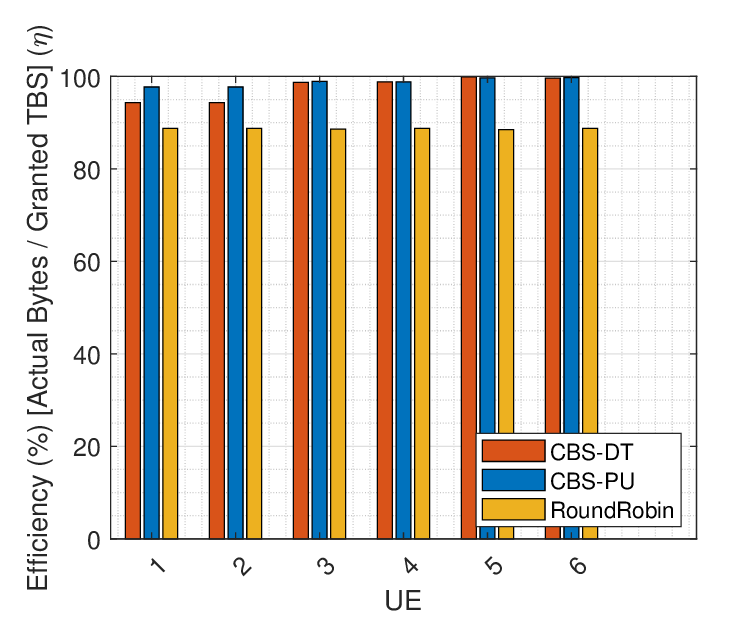} 
        \caption{at $\rho=0.2$} 
        \label{fig:rho20} 
    \end{subfigure}
    \begin{subfigure}[b]{0.48\columnwidth}
        \includegraphics[width=\linewidth]{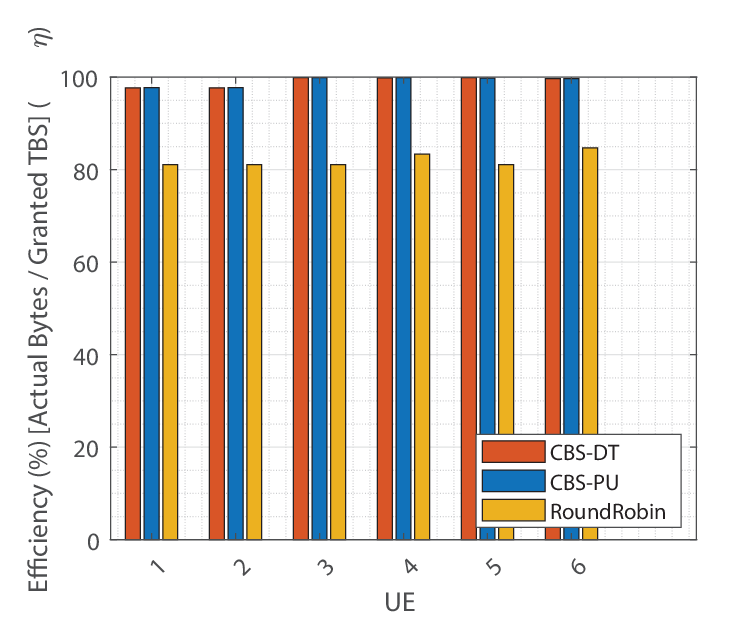} 
        \caption{at $\rho=0.4$} 
        \label{fig:rho40}
    \end{subfigure}
    \caption{Per-UE downlink resource utilization efficiency ($\eta_u$) comparison at reduced load.} 
    \label{fig:utilization_low_load}
\end{figure}
\subsection{Scalability of the Event-Driven Scheduler}
Fig.~\ref{fig:theoretical_scalability} compares a naive per-slot scan,
$T_{\text{slot}}^{\text{naive}} = c_0 + c_U\,U,$ with an event-driven design that updates state only on arrivals and grants, $T_{\text{slot}}^{\text{evt}} = c_0' + c_G\,\bar G + c_H\,(\bar A+\bar G)\log U,$
where $\bar A$ and $\bar G$ are the expected counts per slot of (i) queue-becomes-non-empty events (heap inserts) and (ii) new-grant events (heap pops and possible reinserts), respectively. Here $c_0, c_U, c_0', c_G, c_H > 0$ are implementation-dependent constants (per-slot overheads, per-UE scan cost, per-grant bookkeeping, heap-operation factors). 
This follows from the amortized bound in Section~\ref{sec:complexity},
$O(N{+}G) + O\!\big((A{+}G)\log U\big)$ over $N$ slots, after division by $N$.
The crossover population $U^\star$ solves
$c_0 + c_U U \;=\; c_0' + c_G \bar G + c_H (\bar A+\bar G)\log U,$
with a unique solution when $\bar A$ and $\bar G$ are independent of $U$ (e.g., a per-slot grant cap $K$ keeps $\bar G=O(1)$). In that regime, the event-driven cost grows only logarithmically in $U$. If, instead, $\bar A$ or $\bar G$ scale with $U$  then $T_{\text{slot}}^{\text{evt}}=\Theta\!\big((\bar A{+}\bar G)\log U\big)$ and can become $\Theta(U\log U)$.
\begin{figure}[t!]
  \centering
  \includegraphics[scale = 0.55]{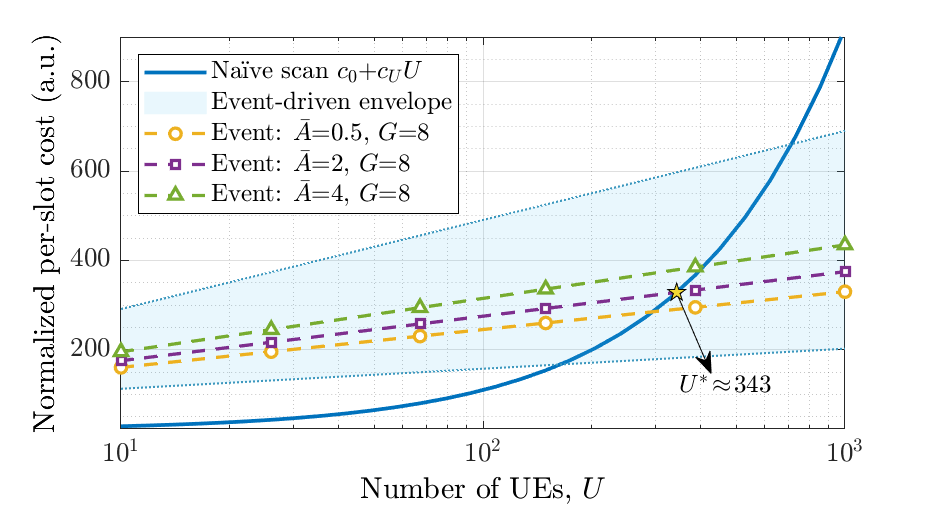}
\caption{Per-slot cost vs.\ UE count. Shaded band: \emph{model envelope} over $(\bar A,\bar G)\!\in\!\{0.5,2,4\}\!\times\!\{4,8,16\}$; not a confidence interval.} 
  \label{fig:theoretical_scalability}
\end{figure}
\section{Conclusion}
\label{sec_conclusion}
This paper addressed the mismatch between continuous-time Ethernet CBS and the slotted, grant-driven operation of 5G/6G networks by introducing a slot-native credit regulator inside the gNB scheduler. We developed two debit rules: CBS-DT, which debits the granted TBS, and CBS-PU, which debits only delivered bytes to avoid over-penalizing partially filled TBs. On the analytical side, we established bounded credit evolution and finite deficit recovery, and---when CBS is layered over an RR selector with at most $K$ new grants per slot---derived deterministic bounds on (i) time to regain eligibility, (ii) time to first service once eligible, and (iii) inter-grant separation. We further gave an event-driven realization that replaces per-slot $\mathcal{O}(U)$ scans with wake-ups, yielding amortized work $\mathcal{O}(N{+}G)+\mathcal{O}((A{+}G)\log U)$ over $N$ slots, with $U$ the UE count, $A$ the number of queue-activation events, and $G$ the number of new-grant events; thus, the cost scales with traffic events rather than population. Our 3GPP-conformant simulations (MATLAB 5G Toolbox) show that both CBS variants preserve strict priority ordering under overload and markedly tighten latency tails relative to PF/WPF. 
\section*{Acknowledgement}
This work was funded by the Dutch "Rijksdienst voor Ondernemend Nederland – RVO" under the TSH number 21007~–~RHIADA project. The authors would also like to acknowledge the project partners, GKN Fokker Aerospace and the Netherlands Aerospace Centre (NLR), for their support and collaboration. 
\bibliographystyle{IEEEtranUrldate}
\bibliography{references}
\end{document}